\documentclass[11pt,onecolumn]{article}
\usepackage{amsmath, amssymb}
\usepackage{amsthm}
\usepackage{float}
\usepackage{enumitem}
\usepackage{color}

\usepackage[linesnumbered,ruled,vlined]{algorithm2e}
\usepackage[noend]{algpseudocode}
\usepackage{graphicx,tikz,subfig}

\usepackage[left=0.9in,top=0.9in,right=0.9in,bottom=0.5in,nohead,textheight=10in,footskip=0.3in]{geometry}
\usepackage{cleveref}

\usepackage[font=small,labelfont=bf]{caption}

\newcommand{\eps}{\varepsilon}

\newtheorem{theorem}{Theorem}
\newtheorem{lemma}[theorem]{Lemma}

\newtheorem{proposition}[theorem]{Proposition}
\newtheorem{corollary}[theorem]{Corollary}

\newtheorem{observation}[theorem]{Observation}

\theoremstyle{definition}

\newcommand{\main}{\text{out}}
\newcommand{\base}{\text{in}}

\DeclareMathOperator{\poly}{poly}

\begin{document}

\title{Algorithms for matrix multiplication via sampling and opportunistic matrix multiplication\footnote{This is an extended version of a paper appearing in the European Symposium on Algorithm (ESA) 2023.}}
\author{David G. Harris\footnote{University of Maryland. Email: {\tt davidgharris29@gmail.com}}}

\maketitle

\abstract{ Karppa \& Kaski (2019) proposed a novel ``broken" or ``opportunistic" matrix multiplication algorithm, based on a variant of Strassen's algorithm, and used this to develop new algorithms for Boolean matrix multiplication, among other tasks. Their algorithm can compute Boolean matrix multiplication in $O(n^{2.778})$ time. While asymptotically faster matrix multiplication algorithms exist, most such algorithms are infeasible for practical problems. 

We describe an alternative way to use the broken multiplication algorithm to approximately compute matrix multiplication, either for real-valued or Boolean matrices. In brief, instead of running multiple iterations of the broken algorithm on the original input matrix, we form a new larger matrix by sampling and run a single iteration of the broken algorithm on it. Asymptotically, our algorithm has runtime $O(n^{2.763})$, a slight improvement over the Karppa-Kaski algorithm.

Since the goal is to obtain new practical matrix-multiplication algorithms, we also estimate the concrete runtime for our algorithm for some large-scale sample problems. It appears that for these parameters, further optimizations are still needed to make our algorithm competitive.}

\section{Introduction}
Consider the fundamental computational problem of \emph{matrix multiplication}: given matrices $A, B$ of dimensions $n_1 \times n_3$ and $n_3 \times n_2$ respectively, our goal is to compute the matrix $C$ given by
$$
C_{i,j} = \sum_k A_{i,k} B_{k,j}
$$

 It has an obvious $O(n_1 n_2 n_3)$ algorithm (the so-called ``naive algorithm''), which simply iterates over all values $i,j,k$. There has been a very long of research into developing asymptotically faster algorithms. For square matrices, we write $n = n_1 = n_2 = n_3$ and the runtime is given as $n^{\omega+o(1)}$, where $\omega$ is the so-called linear algebra constant. Currently, the best bound \cite{wwnewest} is $\omega \leq 2.372$, coming from a variant of Coppersmith-Winograd's algorithm \cite{cite4}. There are also fast algorithms for rectangular matrix multiplication \cite{cite5,cite6,f2,wwnewest}. 

Unfortunately, nearly all of the fast matrix multiplications algorithms are completely impractical due to large hidden constants. Only a small handful of these algorithms are usable \emph{in practice}; by far the most important is Strassen's algorithm and its variants, which have runtime $O(n^{\log_2 7})$.  There is an extensive literature on optimizing and implementing Strassen's algorithm in various computational platforms, see e.g. \cite{cite2, cite9, cite13}.  Depending on the matrix shape, a few rectangular algorithms may also be practical \cite{cite3}.

There is an important special case of \emph{Boolean} matrix multiplication (BMM), where the entries of $A, B$ come from the algebra $\{0,1 \}$ with binary operations $\vee, \wedge$, and our goal is to compute the matrix $C$ given by
$$
C_{i,j} = \bigvee_k A_{i,k} \wedge B_{k,j}
$$

This problem, along with some variants, is a primitive used in algorithms for transitive closures, parsing context-free grammars, and shortest path problems in unweighted graphs, among other applications. Again, the obvious $O(n_1 n_2 n_3)$ naive algorithm can be used. There are a number of other specialized algorithms based on combinatorial optimizations, see e.g. \cite{cite15}.

There is a standard reduction from BMM to integer matrix multiplication: compute the matrix product $\tilde C = A B$ over the integers, and then set $C_{i,j} = 1$ if $\tilde C_{i,j} \geq 1$. Alternatively, there is a randomized reduction from BMM to matrix multiplication over the finite field $\text{GF}(2)$: set each non-zero entry of $B$ to zero with probability $1/2$, and then compute the $\text{GF}(2)$ matrix product $\tilde C = A B$. Then $\tilde C_{i,j} = 1$ with probability $1/2$ whenever $C_{i,j} = 1$, else $\tilde C_{i,j} = 0$ with probability one. With further $O(\log n)$ repetitions the error probability can be reduced to a negligible level. 

It seems difficult to make progress on better practical algorithms for matrix multiplication. In \cite{kk}, Karppa \& Kaski proposed an innovative and novel approach to break this impasse: they described a ``broken'' or ``opportunistic'' form of matrix multiplication, which uses a variant of Strassen's algorithm to compute a tensor which includes a subset of the terms in the full matrix multiplication tensor.  For brevity, we refer to their algorithm as the \emph{KK algorithm}.   By iterating for sufficiently many repetitions, this can be used to solve BMM with high probability.  With appropriate choice of parameters, the overall runtime is $O(n^{\log_2(6+6/7)} \log n) \approx n^{2.776}$, a notable improvement over Strassen. See also \cite{cite13} for further details.

In this paper, we describe an alternative way to use the KK algorithm: instead of executing multiple independent iterations, we combine them all into a single larger randomized broken multiplication. Each term of the original matrix multiplication tensor gets sampled multiple times into the larger matrix. Because matrix multiplication has sub-cubic runtime, it is more efficient to compute this single larger matrix multiplication compared to the multiple smaller matrices.

We get the following main result (given here in a slightly simplified form):

\begin{theorem}
For square matrices of dimension $n$, we can compute BMM in $O( n^{\frac{3 \log 6}{\log 7}} (\log n)^{\frac{\log 6}{\log 7}}) \leq O(n^{2.763})$ bit operations with high probability.
\end{theorem}

We can also use this method to obtain a randomized approximation for real-valued matrix multiplication. The precise approximation guarantees depend on the values of the input matrices. When the input matrices $A, B$ (and hence $C$) are non-negative, then we have the following crisp statement (again, in simplified form):
\begin{theorem}
For non-negative square matrices of dimension $n$, we can obtain an approximating matrix $\tilde C$ in $O(\frac{n^{2.763}}{\eps^2})$ operations, with $\tilde C_{i,j} \in [(1-\epsilon), (1+\epsilon)] C_{i,j}$ with high probability for all $i,j$.
\end{theorem}

For example, in the setting of Boolean matrix multiplication, we can not only \emph{detect} if there is any value $k$ with $A_{i,k} = B_{k,j} = 1$, but we can \emph{estimate} the number of such witnesses $k$, up to any desired relative accuracy.

The new algorithm is quite simple, and also has a number of advantages from the viewpoint of practical implementations. First, it lends itself easily to rectangular (non-square) input matrices; we will show a more general theorem which describes the runtime scaling for matrices of arbitrary shape. Second, the algorithm itself internally uses square matrices, and is flexible about the precise dimensions used; this avoids a number of tedious issues with padding in matrix algorithms.

We emphasize that this paper is structured somewhat differently from a typical paper in theoretical computer science. Our goal is to develop a new \emph{practical} algorithm (we already have good theoretical algorithms!), so we cannot afford to look purely at asymptotic estimates. Instead, we have developed our formulas and calculations much more preicsely to get concrete runtime estimates in addition to the usual asymptotic bounds.

In Section~\ref{sec:practical}, we illustrate with some computational estimates for large problem instances which are at the limit of practicality. Unfortunately, as we will see, the asymptotic behavior does not seem to fully kick in even at such large scales. At the current time, it appears that the new algorithm is unlikely to beat alternative algorithms; further optimizations and improvements may be needed.

We note that \cite{kk} also discussed generalizations to other types of broken matrix multiplication tensors. We will restrict our attention to Strassen-based pseudo-multiplication, for two main reasons. First, unlike in \cite{kk},  our analysis heavily depends on specific structural properties of the Strassen tensor as opposed to just counting the gross number of terms. Second, it is not clear if any other tensors are practical, especially in light of the fact that Strassen's algorithm appears to be the only truly practical fast multiplication algorithm.

\section{Pseudo-multiplication}
To describe our BMM algorithm, we need to discuss the underlying broken multiplication algorithm in depth. We call this procedure ``pseudo-multiplication'' or ``pseudo-product'' to distinguish it from the overall KK algorithm itself.\footnote{ \cite{kk} includes an additional step, where the entries of the matrices are randomly permuted at each level. We omit this, since we will later include more extensive randomization in the overall algorithm.}  We begin by quoting the algorithmic result of \cite{kk}.
\begin{theorem}[\cite{kk}]
\label{kk-thm1} Consider $2 \times 2$ matrices $A, B$ over an arbitrary ring. Using $14$ additions and $6$ multiplications, we can compute the following quantities:
\begin{align*}
C_{11} &= A_{12} B_{21} & 
C_{21} &= A_{21} B_{11} + A_{22} B_{21} \\
C_{12} &= A_{11} B_{12} + A_{12} B_{22} &
C_{22} &= A_{21} B_{12} + A_{22} B_{22}
\end{align*}
\end{theorem}
This differs from the computation of $C = A B$, in that the term $C_{11}$ is missing the summand $A_{11} B_{11}$. This computation is achieved by a variant of a Strassen step, except that one of the 7 multiplications is omitted.

Scaling up, consider integer matrices $A, B$ of dimension $m = 2^s$. We can identify the integers in the range $[m] = \{0, \dots, m-1 \}$ with binary vectors $\{0,1 \}^s$; thus, we may write $A_{x,y}$ for vectors $x, y \in \{0, 1 \}^s$. We write $z = x \vee y$ where $z_i$ is zero iff $x_i = y_i = 0$, i.e. the disjunction is done coordinate-wise. By iterating the algorithm of Theorem~\ref{kk-thm1} for $s$ levels, we get the following:
\begin{theorem}
\label{kk-thm2}
Using $7(6^s - 4^s)$ additions and $6^s = m^{\log_2 6}$ multiplications, we can compute the matrix pseudo-product $C = A \boxtimes B$ over an arbitrary ring, defined as follows:
$$
C_{x,y} = \sum_{x \vee y \vee z = \vec 1^s} A_{x,z} B_{z,y}
$$
where $\vec 1^s$ denotes the $s$-dimensional vector whose entries are all equal to $1$.
\end{theorem}
\begin{proof}
The arithmetic cost was shown already in \cite{kk}. We show the formula for $C$ by induction on $s$; the base case $s = 0$ holds vacuously.  For the induction step, write $A, B, C$ as $2 \times 2$ block matrices $\bar A, \bar B, \bar C$, whose entries are matrices of dimension $2^{s-1}$.  By iterating Theorem~\ref{kk-thm1}, we have
\begin{align*}
\bar C_{11} &= \bar A_{12} \boxtimes \bar B_{21} &
\bar C_{12} &= \bar A_{11} \boxtimes \bar B_{12} + \bar A_{12} \boxtimes \bar B_{22} \\
\bar C_{21} &= \bar A_{21} \boxtimes \bar B_{11} + \bar A_{22} \boxtimes \bar B_{21} &
\bar C_{22} &= \bar A_{21} \boxtimes \bar B_{12} + \bar A_{22} \boxtimes \bar B_{22}
\end{align*}

Consider some entry $x,y \in \{0,1 \}^s \times \{0,1 \}^s$ with $x_1 + y_1 \geq 1$; these correspond to submatrices $\bar C_{12}$ or $\bar C_{21}$ or $\bar C_{22}$. For concreteness, let us suppose that $x_1 = 1, y_1 = 0$, corresponding to $\bar C_{21}$; the other cases are completely analogous. We can write $x = 1 | x', y = 0 | y'$ where $x', y' \in \{0, 1 \}^{s-1}$ and the notation $1 | x'$ denotes concatentation. By induction hypothesis, we have
$$
(\bar A_{21} \boxtimes \bar B_{11})_{x', y'}  = \sum_{x' \vee y' \vee z' = \vec 1^{s-1}} (\bar A_{21})_{x', z'} (\bar B_{11})_{z', y'} = \sum_{x' \vee y' \vee z' = \vec 1^{s-1}} A_{1|x',0|z'} B_{0|z',0|y'}.
$$

A similar formula holds for the other term $\bar A_{22} \boxtimes \bar B_{21}$, and so 
$$
C_{x,y} =  \sum_{x' \vee y' \vee z' = \vec 1^{s-1}} A_{1|x',0|z'} B_{0 | z',0 | y'} + \sum_{x' \vee y' \vee z' = \vec 1^{s-1}} A_{1 | x',1 | z'} B_{0 | z',0 | y'}.
$$

Since $x_1 = 1$, we have $x' \vee y' \vee z' = \vec 1^{s-1}$ iff $x \vee y \vee z = \vec 1^s$, and this is precisely $C_{xy} = \sum_{x \vee y \vee z = \vec 1^s} A_{x,z} B_{z,y}$ as claimed.

Next, consider the entry $x,y$ with $x_1 = y_1 = 0$, corresponding to $\bar C_{11}$. We can write $x = 0|x', y = 0|y'$. By induction hypothesis, we have
$$
(\bar A_{11} \boxtimes \bar B_{11})_{y', z'} = \sum_{x' \vee y' \vee z' = \vec 1^{s-1}} (\bar A_{11})_{x' , z'} (\bar B_{11})_{z', y'} = \sum_{x' \vee y' \vee z' = \vec 1^{s-1}} A_{0|x',1|z'} B_{1|z',0|y'}.
$$
So $C_{x,y} = \sum_{x' \vee y' \vee z' = \vec 1^{s-1}} A_{x, 1|z'} B_{1|z',y} = \sum_{x \vee y \vee z =  \vec 1^s} A_{x,y} B_{y,z}$ as desired.
\end{proof}

In particular, this implies the following result (which was the only thing used directly in \cite{kk}):
\begin{corollary}
The pseudo-product $C = A \boxtimes B$ contains $(7/8)^s$ of the summands in the full matrix product $C = A B$.
\end{corollary}
\begin{proof}
There are precisely $7^s$ triples $(x,y,z) \in \{0, 1 \}^{3 s}$ with $x \vee y \vee z = 7^s$.
\end{proof}

\medskip

\noindent \textbf{Recursion base case.} In practice, with recursive matrix decompositions such as the KK algorithm or Strassen's algorithm, it is not optimal to follow the recursion all the way down to the underlying ring. Instead, one should switch over to a direct ``base case'' implementation of matrix multiplication for the innermost calculations.  For the KK algorithm, unlike for ordinary matrix multiplication, this is not simply a matter of computational convenience: when we switch to the base case, we compute the matrix multiplication \emph{exactly}, and so we will need to take this into account both for analyzing the algorithm accuracy as well as its runtime.

 Specifically, let us suppose that we use $s$ steps of pseudo-multiplication, followed by true matrix multiplication on the resulting submatrices of dimension $d_i$ (the \emph{base case}).  Let $m_i = d_i 2^s$ for $i = 1,2,3$. We can view any integer $x \in \{1, \dots, m_i \}$ as equivalent to an ordered pair $(x^\main, x^\base)$ where $x^\main \in \{0,1 \}^s$ and $x^\base \in \{1, \dots, d_i \}$.  When $s$ is understood, we define $V = \{0,1 \}^s$ and define $T$ to be the set of triples $(x,y,z) \in V^3$ with $x \vee y \vee z = \vec 1^s$.   Furthermore, we define $T^*$ to be the set of triples $(x^\main, x^\base), (y^\main, y^\base), (z^\main, z^\base)$ with $(x^\main,y^\main,z^\main) \in T$ (where again $s$ and the base dimensions $d_i$ are all understood). 
 
  With these conventions, the resulting matrix after the outer  $s$ recursive levels of pseudo-multiplication and the inner base case computation is then given by:
$$
C_{x,y} = \sum_{z: (x,y,z) \in T^*} A_{x,z} B_{z,y} = \sum_{\substack{z = (z^\main, z^\base) \in \{0, 1 \}^s \times \{1, \dots, d_3 \} \\ x^\main \vee y^\main \vee z^\main = \vec 1^s}} A_{(x^\main, x^\base), (z^\main, z^\base)} B_{(z^\main, z^\base), (y^\main, y^\base)}.
$$
 
Over the field $\text{GF}(2)$, the base case would likely involve using the Four Russians method \cite{fourrussians} and/or bitsliced operations. These optimizations are very powerful in practice, possibly more important than the asymptotic gains from a fast matrix multiplication algorithm. As an example, in the implementation of \cite{cite10}, they choose $d_i \approx 2^{11}$ (roughly matching the L2 cache size). It may be convenient to have $d_3 \gg d_1 = d_2$ in order to a greater degree of SIMD processing.  From a theoretical point of view, the base dimensions $d_i$ may be viewed as a constant.   

\section{Algorithm overview}
Consider matrices $A, B$ of dimension $n_1 \times n_3$ and $n_3 \times n_2$; here the matrices are either real-valued or Boolean-valued, and our goal is to approximate the product $C = AB $. The algorithms we use are very similar for these two cases, as we summarize in the following pseudocode.

\begin{algorithm}[H]
\DontPrintSemicolon
choose parameter $s$ and base sizes $d_1, d_2, d_3$ \\
\For{$ i = 1,2,3$} {
  set $m_{i} = 2^s d_{i}$ and  draw random functions $F_{i}: [m_{i}] \rightarrow [n_{i}]$. 
}
form matrices $\bar A, \bar B$ of dimension $m_1 \times m_3$ and $m_3 \times m_2$ by sampling:
$$
\bar A_{x,z} = G^A_{x,z} \cdot A_{F_1(x), F_3(z)} , \qquad \bar B_{z,y} = G^B_{z,y} \cdot B_{F_3(z), F_2(y)} 
$$
 for scaling matrices $G^A, G^B$. \\
compute pseudo-multiplication $\bar C = \bar A \boxtimes \bar B$ \\
\For{each entry $i,j \in [n_1] \times [n_2]$} {
  obtain entry $\tilde C_{i,j}$ by aggregating matrix entries $\bar C_{x,y}: x \in F_1^{-1}(i), y \in F_2^{-1}(j)$.
}
\caption{The matrix multiplication algorithm}
\end{algorithm}

For the implementation of steps 4 and 6 (the scaling and aggregation steps), we do the following.

\medskip

\noindent \textbf{Real-valued algorithm.} Here, we choose non-negative real-valued scaling matrices $G^A, G^B, G^C$ of dimensions $m_1 \times m_3, m_3 \times m_2, m_1 \times m_2$ respectively. We then define
$$
\bar A_{x,z} = G_{x,z}^A \cdot A_{F_1(x), F_3(z)}, \qquad \bar B_{z,y} = G^B_{z,y} \cdot B_{F_3(z) F_2(y)}
$$
and, after computing the pseudo-product $\bar C = \bar A \boxtimes \bar B$ over $\mathbb R$, we estimate $C$ by
$$
\tilde C_{i,j} = \sum_{x \in F_1^{-1}(i), y \in F_2^{-1}(j)} G^C_{x,y} \bar C_{x,y}.
$$

\noindent \textbf{Boolean algorithm}. Here, in addition to the sampling functions $F$, we also draw a uniformly random $m_3 \times m_2$ Boolean matrix $D$. We then form matrices $\bar A, \bar B$ by setting:
$$
\bar A_{x,z} = A_{F_1(x),F_3(z)}, \qquad \bar B_{z,y} = B_{F_3(z),F_2(y)} \wedge D_{z,y}.
$$

We compute the $\bar C = \bar A \boxtimes \bar B$ over $\text{GF}(2)$ and produce an estimated matrix $\tilde C$ by
$$
\tilde C_{i,j} = \bigvee_{x \in F_1^{-1}(i), y \in F_2^{-1}(j)} \bar C_{x,y}.
$$

\bigskip

In either case, this generic algorithm has runtime of $O(6^s d_1 d_2 d_3)$ ring operations and uses $O(4^s (d_1 d_3 + d_2 d_3 + d_1 d_2))$ memory to store $\bar A, \bar B, \bar C$.

Before explaining the algorithm details, let us provide an intuition for when this process should work. The tensor for the matrix product $\bar A \bar B$ contains $8^s d_1 d_2 d_3$ terms $A_{x,z} B_{z,y}$. When we compute the pseudo-product $\bar A \boxtimes \bar B$, we are effectively subsampling a $(7/8)^s$ fraction of these terms, and so overall it contains $7^s d_1 d_2 d_3$ tensor terms.  On the other hand, the desired exact matrix product $A B$ contains $n_1 n_2 n_3$ terms. If the pseudo-multiplication tensor were to sample a \emph{random} subset of $(7/8)^s$ of the tensor terms, then as long as $$
7^s d_1 d_2 d_3 \gg n_1 n_2 n_3,
$$
 every term of the original matrix tensor would be sampled with good probability. In this case, we would expect that $\bar A \boxtimes \bar B$ contains complete information about the product $C = A B$ and so our approximated matrix $\tilde C$ should be close to $C$.

What makes this complicated is that the pseudo-multiplication tensor is heavily biased; specifically, terms $\bar A_{x,z} \bar B_{z,y}$ are over-represented when the Hamming weights of the vectors $x,y,z$ are large. It requires significant analysis and the use of judiciously chosen correction terms to deal with this bias. This was not an issue encountered in the original analysis \cite{kk}, which only worked expectation-wise.

 We will analyze the algorithms in three main steps. First, in Sections~\ref{sec:analysis1} and \ref{sec:analysis2}, we derive some explicit non-asymptotic bounds on the accuracy. In Section~\ref{sec:asymptotic} we use these to show bounds on the asymptotic algorithm complexity. These too crude for practical-scale computations; in Section~\ref{sec:practical}, we finish with some cases studies on very large problem parameters pushing the limits of practical computations, with concrete bounds.

\section{Analysis: the real-valued case}
\label{sec:analysis1}
Our strategy to analyze $\tilde C_{i,j}$ is to compute the mean and variance. We will show that $\mathbb E[\tilde C_{i,j}] \propto C_{i,j}$ (with known proportionality constants). Thus, by rescaling the matrix entries $G$ as needed, we can obtain an unbiased estimate of $C_{i,j}$. Likewise, the variance of $\tilde C_{i,j}$ is determined by $C_{i,j}$ as well as its ``second moment'' matrix
$$
C_{i,j}^{(2)} = \sum_k A_{i,k}^2 B_{k,j}^2.
$$

\begin{lemma}
\label{prop1}
For any $(x,y,z) \in T^*$, define $G_{xyz} = G^A_{x,z} G^B_{y,z} G^C_{x,y}$. Define the sums $\sigma_0, \dots, \sigma_{123}$ by:
\begin{align*}
\sigma_0 &= \sum_{(x,y,z) \in T^*} G_{xyz} & 
\sigma_1 &= n_1 \sum_{x} \Bigl( \sum_{y,z: (x,y,z) \in T^*} G_{xyz} \Bigr)^2 \\
\sigma_2 &= n_2 \sum_{y} \Bigl( \sum_{x,z: (x,y,z) \in T^*} G_{xyz} \Bigr)^2 &
\sigma_3 &= n_3 \sum_{z } \Bigl( \sum_{x,y: (x,y,z) \in T^*} G_{xyz} \Bigr)^2 \\
\sigma_{12} &= n_1 n_2\sum_{ x, y }\Bigl( \sum_{z: (x,y,z) \in T^*} G_{xyz} \Bigr)^2 &
\sigma_{13} &= n_1 n_3  \sum_{ x, z } \Bigl( \sum_{y: (x,y,z) \in T^*} G_{xyz} \Bigr)^2 \\
\sigma_{23} &= n_2 n_3 \sum_{ y, z} \Bigl( \sum_{x: (x,y,z) \in T^*} G_{xyz} \Bigr)^2 &
\sigma_{123} &= n_1 n_2 n_3 \sum_{(x,y,z) \in T^*} G_{xyz}^2
\end{align*}

Then, for any $i,j$, there holds $$
\mathbb E[ \tilde C_{i,j} ] = \frac{\sigma_0 C_{i,j}}{n_1 n_2 n_3}, \qquad \text{and} \qquad \mathbb V[ \tilde C_{i,j} ] \leq \frac{(\sigma_1 + \sigma_2 + \sigma_{12}) C_{i,j}^2 + (\sigma_3 + \sigma_{13} + \sigma_{23} +  \sigma_{123})  C_{i,j}^{(2)}}{(n_1 n_2 n_3)^2} 
$$
\end{lemma}
\begin{proof}
For each entry $i,j$, we have 
$$
\tilde C_{i,j} = \sum_k \sum_{(x,y,z) \in T^*}  G_{xyz} [F_1(x) = i] [F_2(y) = j] [F_3(z) = k ] A_{i,k} B_{k,j}
$$
where we use Iverson notation for any Boolean predicate here and throughout the paper, i.e. $[E] = 1$ if $E$ holds, otherwise $[E] = 0$. 

Clearly $\Pr(F_1(x) = i) = 1/n_1$ and $\Pr(F_2(y) = j) = 1/n_2$ and $\Pr(F_3(z) = k) = 1/n_3$, so
$$
\mathbb E[C_{i,j}] = \sum_k \sum_{(x,y,z) \in T^*} \frac{G_{xyz}}{n_1 n_2 n_3} A_{i,k} B_{k,j} = \frac{\sigma_0}{n_1 n_2 n_3} \sum_k A_{i,k} B_{k,j} = \frac{\sigma_0}{n_1 n_2 n_3} C_{i,j}
$$

Next, we can write
\begin{align}
\mathbb E[ \tilde C_{i,j}^2 ] &= \negthickspace \negthickspace \sum_{\substack{(x,y,z) \in T^*, k \in [n_3] \\ (x', y', z') \in T^*, k' \in [n_3]}} \negthickspace \negthickspace \Pr(F_1(x) = F_1(x') = i \wedge F_2(y) = F_2(y') = j \wedge F_3(z) = k \wedge F_3(z') = k') \notag \\[-0.20in] 
 & \qquad \qquad \qquad \qquad \qquad \qquad \cdot A_{i,k} B_{k,j} A_{i,k'} B_{k',j} G_{xyz} G_{x'y'z'} \label{ssum0}
\end{align}

We divide the sum in (\ref{ssum0}) into two parts. First, consider pairs $(x,y,z), (x',y',z')$ with $z = z'$ and $k = k'$; they have a total contribution of
\begin{align}
&\sum_{\substack{(x,y,z) \in T^*, k \in [n_3] \\ (x', y', z) \in T^*}} \negthickspace \negthickspace \Pr(F_1(x) = F_1(x') = i \wedge F_2(y) = F_2(y') = j) \frac{ A_{i,k} B_{k,j} A_{i,k} B_{k,j} G_{xyz} G_{x'y'z} }{n_3} \notag \\
&=\sum_{\substack{(x,y,z) \in T^*, (x', y', z) \in T^*}} \negthickspace \negthickspace \Pr(F_1(x) = F_1(x') = i \wedge F_2(y) = F_2(y') = j) \sum_{k} \frac{ A_{i,k}^2 B_{k,j}^2  G_{xyz} G_{x'y'z} }{n_3} \notag \\
&=\sum_z \sum_{\substack{x,x', y,y' \\ (x,y,z) \in T^*, (x', y', z) \in T^*}} \Pr(F_1(x) = F_1(x') = i \wedge F_2(y) = F_2(y') = j)  \frac{C_{i,j}^{(2)} G_{xyz} G_{x'y'z} }{n_3} \label{ssum1}
\end{align}

In the RHS of (\ref{ssum1}),  we consider four separate cases depending on whether $x = x'$ and/or $y = y'$. The case with $x = 
x', y = y'$ would contribute
\begin{align*}
\sum_z \sum_{\substack{(x,y) \in T^*}} \Pr(F_1(x) = i \wedge F_2(y) = j)  \frac{C_{i,j}^{(2)} G_{xyz} G_{xyz} }{n_3} 
&= \sum_{(x,y,z) \in T^*}  \frac{C_{i,j}^{(2)} G_{xyz}^2 }{n_1 n_2 n_3} = \frac{\sigma_{123}}{(n_1 n_2 n_3)^2} C_{i,j}^{(2)}
\end{align*}
and similarly the case with $x = x', y \neq y'$ would contribute
\begin{align*}
&\sum_z \sum_{\substack{x,y \neq y' \\ (x,y,z) \in T^*, (x, y',z) \in T^*}} \Pr(F_1(x) = i \wedge F_2(y) = F_2(y') = j)  \frac{C_{i,j}^{(2)} G_{xyz} G_{xy'z} }{n_3} \\
&\qquad = \sum_{x,z} \sum_{\substack{y \neq y' \\ (x,y,z) \in T^*, (x,y',z) \in T^*}} \frac{ C_{i,j}^{(2)} G_{xyz} G_{x y'z} }{n_1 n_2^2 n_3} \leq  \sum_{x,z} \sum_{\substack{y,y' \\ (x,y,z) \in T^*, (x,y',z) \in T^*}} \frac{ C_{i,j}^{(2)} G_{xyz} G_{x y'z} }{n_1 n_2^2 n_3} \\
&\qquad = \sum_{x,z} \Bigl( \sum_{\substack{y: (x,y,z) \in T^*}} G_{xyz} \Bigr)^2 \cdot \frac{C_{i,j}^{(2)}}{n_1 n_2^2 n_3} = \frac{\sigma_{13}}{(n_1 n_2 n_3)^2} C_{i,j}^{(2)}
\end{align*}

In completely analogous ways, the cases with $x \neq x', y = y'$ and $x \neq x', y \neq y'$ would contribute $\frac{\sigma_{23}}{(n_1  n_2  n_3)^2} C_{i,j}^{(2)}$ and $\frac{\sigma_3}{(n_1 n_2 n_3)^2} C_{i,j}^{(2)}$ respectively.

Next, consider pairs $(x,y,z), (x',y',z')$ in (\ref{ssum0}) with $z \neq z'$. Then $\Pr(F_3(z) = k,  F_3(z') = k') = \frac{1}{n_3^2}$ so the total contribution of these terms is given by
\begin{align*}
&\sum_{\substack{(x,y,z) \in T^*, (x', y', z') \in T^* \\ z \neq z'}}  \Pr(F_1(x) = F_1(x') = i \wedge F_2(y) = F_2(y') = j) G_{xyz} G_{x'y'z'} \sum_{k,k'} \frac{ A_{i,k} B_{k,j} A_{i,k'} B_{k',j}} {n_3^2} \\
&\leq \frac{C_{i,j}^2} {n_3^2} \sum_{\substack{(x,y,z) \in T^* \\  (x', y', z') \in T^*}} \Pr(F_1(x) = F_1(x') = i \wedge F_2(y) = F_2(y') = j) G_{xyz} G_{x'y'z'} 
 \end{align*}
 
 Again, we can divide this sum into four different cases depending on whether $x = x'$ and/or $y = y'$. The case with $x = x', y = y'$ would contribute 
 \begin{align*}
& \frac{C_{i,j}^2} {n_3^2} \sum_{\substack{(x,y,z) \in T^* \\ (x,y, z') \in T^*}} \frac{ G_{xyz} G_{xyz'} }{n_1 n_2} = \frac{C_{i,j}^2} {n_1 n_2 n_3^2} \sum_{x,y} \Bigl( \sum_{z: (x,y,z) \in T^*} G_{xyz} \Bigr)^2 = \frac{\sigma_{12}}{ (n_1 n_2 n_3)^2 } C_{i,j}^2
 \end{align*}
 and, in a completely symmetric way, the cases with $x \neq x', y = y'$ and $x = x', y \neq y'$ and $x \neq x', y \neq y'$ would contribute respectively $\frac{C_{i,j}^2}{ (n_1 n_2 n_3)^2} \sigma_2$ and $\frac{C_{i,j}^2}{ (n_1 n_2 n_3)^2} \sigma_1$ and $\frac{C_{i,j}^2}{ (n_1 n_2 n_3)^2} \sigma_0^2$.
\end{proof}

For non-negative matrices, we get the following crisp formulation:
\begin{corollary}
If matrices $A, B$ are non-negative, then 
$$
\frac{ \mathbb V[ \tilde C_{i,j} ]}{ \mathbb E[ \tilde C_{i,j} ]^2} \leq \frac{ \sigma_1 + \sigma_2 + \sigma_3 + \sigma_{12} + \sigma_{13} + \sigma_{23} + \sigma_{123} }{\sigma_0^2} 
$$
\end{corollary}
\begin{proof}
We know $C_{i,j} \geq 0$ and $C_{i,j}^{(2)} \leq C_{i,j}^2$ for each entry $i,j$. 
\end{proof}

\section{Analysis: the Boolean case}
\label{sec:analysis2}
Here we will argue that, with high probability, we have $\tilde C_{i,j} = C_{i,j}$ for each term $i,j$. We begin with the following easy observation that the algorithm has one-sided error.

\begin{observation}
If $C_{i,j} = 0$, then $\tilde C_{i,j} = 0$.
\end{observation}
\begin{proof}
Consider an entry $\bar C_{x,y}$ for $x \in F^{-1}_1(i), y \in F^{-1}_2(j)$. We have $\bar C_{x,y} = \sum_{ (x,y,z) \in T^*} \bar A_{x,z} \bar B_{z,y} D_{z,y}$. For each summand $z$, we have $\bar A_{x,z} \bar B_{z,y} = A_{i,k} B_{k,j}$ where $k = F_3(z)$. Since $C_{i,j} = 0$, all such terms are zero. Hence $\bar C_{x,y} = 0$. Since this holds for all such $x,y$, we have $\tilde C_{i,j} = 0$ as well.
\end{proof}

So consider a pair $i,j$ with $C_{i,j} = 1$ and witness value $k$ with $A_{i,k} = B_{k,j} = 1$. We next argue that, for each triple $(x,y,z) \in T^*$ which is mapped to $i,j,k$, the randomization coming from the $D$ matrix ensures there is a probability of $1/2$ of getting $\bar C_{x,y} = 1$ and hence $\tilde C_{i,j} = 1$. 
\begin{theorem}
\label{thm8}
Fix an arbitrary value $k$ with $A_{i,k} = B_{k,j} = 1$.  For each $y \in [n_3]$, let $J_y$ be an independent Bernoulli-$1/2$ random variable. For each tuple $t = (x^{\main},y^{\main},z^{\main}) \in T$, define \begin{align*}
I_{1,x^{\main}} &=\bigvee_{x^{\base} \in [d_1]} F_1(x^{\main, \base}) = i \\
 I_{2,y^{\main}} &= \bigvee_{y^{\base} \in [d_2]} F_2(y^{\main}, y^{\base}) = j \wedge J_{(y^{\main}, y^{\base})} = 1  \\
 I_{3,z^{\main}} &=  \bigvee_{z^{\base} \in [d_3]} F_3(z^{\main}, z^{\base}) = k
\end{align*}

Then $\Pr( \tilde C_{i,j}) \leq \Pr(U)$ for the event $U$ defined as:
$$
U = \bigwedge_{ (x^{\main},y^{\main},z^{\main})  \in T} I_{1,x^{\main}} = 0 \vee I_{2,y^{\main}} = 0 \vee I_{3,z^{\main}} = 0
$$
\end{theorem}
\begin{proof}
Suppose we condition on the functions $F$. Let $P$ denote the set of triples $(x,y,z) \in T^*$ with $F_1(x) = i, F_2(y) = j, F_3(z) = k$ and from this, let $S$ denote the number of \emph{distinct} values $y$ encountered, i.e. $S = | \{y : (x,y,z) \in P \} |$. It can be immediately seen that $U$ holds if and only if $J_{y} = 0$ for all tuples $(x,y,z) \in P$, which has probability precisely $2^{-S}$. So $\Pr( U \mid F) = 2^{-S} $. 

Now suppose $P$ contains triples $(x_1, y_1, z_1), \dots, (x_S, y_S, z_S)$ with $y_1, \dots, y_S$ distinct. Each entry $\bar C_{x_{\ell}, y_{\ell} }$ is then obtained by XORing $D_{z_{\ell},y_{\ell}}$ with some other entries $D_{z', y_{\ell}}$. These are all independent random bits, and hence this XOR is equal  to 1 with probability precisely $1/2$. Since $y_1, \dots, y_S$ are distinct, all such events are independent.  Thus, we have
$$
\Pr( \tilde C_{i,j} = 0 \mid F) \leq \Pr( \bar C_{x_1, y_1} = \dots = \bar C_{x_S, y_S} = 0 ) \leq 2^{-S}.
$$

Combining inequalities and integrating over $F$ gives $\Pr( \tilde C_{i,j} = 0) \leq \mathbb E[2^{-S}] = \Pr ( U )$.
\end{proof}

Crucially, the events of \Cref{thm8} have a special ``monomial'' form. To estimate these probabilities, we will use an inequality of Janson  \cite{cite11} which we summarize it in the following form:
\begin{theorem}[Janson's inequality \cite{cite11}]
\label{janson-thm}
Suppose that $X_1, \dots, X_N$ are independent Bernoulli variables, and $\mathcal E$ is a collection of conjunction events over ground set $[N]$, i.e. each event $E \in \mathcal E$ has the form $\bigwedge_{r \in R_E} X_r$ for a given subset $R_E \subseteq [n]$.  For events $E, E' \in \mathcal E$, define $E \sim E'$ if $R_E \cap R_{E'} \neq \emptyset$, and let $\bar E$ be the complement of $E$ (i.e. that event $E$ does not hold.)

Then $\Pr( \bigwedge_{E \in \mathcal E} \bar E) \leq e^{-\kappa}$ for $
\kappa = \sum_{E \in \mathcal E} \mathbb E \Bigl[ \frac{ [E] }{ \sum_{E' \sim E} [E']} \Bigr].
$ (Recall that $[E]$ and $[E']$ are Iverson notations for the events $E, E'$.)
\end{theorem}

To provide some intutition for Theorem~\ref{janson-thm}, note that if the events in $\mathcal E$ were all independent, then $\kappa = \mu := \sum_{E \in \mathcal E} \Pr(E)$, and the probability that none of them hold would be $\prod_{E \in \mathcal E} (1 - \Pr(E)) \leq  e^{-\mu}$. When events in $\mathcal E$ are dependent, $\kappa$ takes into account the average ``clumping size'' of events of $\mathcal E$, and discounts $\mu$ accordingly.

\begin{observation}
\label{thm8c}
The probability of the event $U$  in Theorem~\ref{thm8} can be modeled in terms of Theorem~\ref{janson-thm}, with the following collection of events $\mathcal E $:
$$
E_t \equiv (I_{x,1} = 1 \wedge I_{y,2} = 1\wedge I_{z,3} = 1) \qquad \qquad \text{for $t = (x,y,z) \in T$}
$$

Here  $E_t \sim E_{t'}$ for tuples $t = (x,y,z), t' = (x',y',z')$ if and only if $x = x'$ or $y = y'$ or $z = z'$.

Each underlying variable $I_{i,v}$ is an independent Bernoulli with mean $p_{i}$ given by:
$$
p_1 = 1 - (1 - 1/n_1)^{d_1} \approx \frac{d_1}{n_1}, \quad p_2 = 1 - (1 - 0.5/n_2)^{d_2} \approx \frac{d_2}{2 n_2}, \quad p_3 = 1 - (1 - 1/n_3)^{d_3} \approx \frac{d_3}{n_3}.
$$

We have $\Pr(E_t) = p_1 p_2 p_3$ for all $t \in T$ and $\mu =  7^s p_1 p_2 p_3$.
\end{observation}

\section{Asymptotic analysis}
\label{sec:asymptotic}
For this section,  we let $d_1 = d_2 = d_3 = 1$, and we define parameters
$$
\psi_1 = n_1 + n_2 + n_3, \quad \psi_2 = n_1 n_2 + n_2 n_3 + n_1 n_3, \quad \psi_3 = n_1 n_2 n_3.
$$

For a vector $v \in V$, we define $|v|$ to be the Hamming weight $|v| = \sum_i v_i$. 

For both the real-valued and Boolean algorithms, the key to the analysis is to count the number of triples $(x,y,z) \in T$ with $F_1(x) = i, F_2(y) = j, F_3(z) = k$ for given values $i,j,k$. Our basic strategy is to restrict the analysis to triples where the Hamming weights of $x,y,z$ are all ``average''. Formally, let us define $H \subseteq T$ to be the set of triples $(x,y,z) \in T$ satisfying the following conditions:
\begin{itemize}
\item $|x| \leq 4s/7$ and $|y| \leq 4s/7$ and $|z| \leq 4s/7$
\item $|x \vee y| \leq 6s/7$ and $|x \vee z| \leq 6s/7$ and $|y \vee z| \leq 6s/7$
\end{itemize}

\begin{lemma}
\label{lem2}
There holds $|H| = \Theta(7^s)$.
\end{lemma}
\begin{proof}
There are precisely $7^s$ triples $(x,y,z)$ in $T$. Consider the uniform distribution on $T$; equivalently, for each coordinate $i$, the values $(x_i, y_i, z_i)$ are chosen uniformly at random among the $7$ non-zero values. We need to show that the Hamming weight conditions on the densities of $x,y,z$ are satisfied with constant probability.

For each coordinate $i$, consider the vector $$
R_i = (x_i - 4/7, y_i - 4/7, z_i - 4/7, x_i \vee y_i - 6/7, x_i \vee z_i - 6/7, y_i \vee z_i - 6/7)
$$ and let $R = \sum_i R_i$. Note that $(x,y,z) \in H$ if and only if all the entries of $R$ are non-positive. For this, we use the multidimensional Central Limit Theorem. Since the variables $R_i$ are i.i.d. random vectors with mean zero, the scaled value  $R/\sqrt{s}$ converges as $s \rightarrow \infty$ to a 6-dimensional normal distribution $N(0, \Sigma)$ with covariance matrix given by
$$
\Sigma = \frac{1}{49} 
\begin{bmatrix}
12 & -2 & -2 & 4 & 4 & -3  \\
-2 & 12 & -2 & 4 & -3 & 4 \\
-2 & -2 & 12 & -3 & 4 & 4 \\
4 & 4 & -3 & 6 & -1 & -1 \\
4 & -3 & 4 & -1 & 6 & -1 \\
-3 & 4 & 4 & -1 & -1 & 6
\end{bmatrix}
$$

Note that $\Sigma$ is non-singular. Hence, for the corresponding distribution $N(0, \Sigma)$ there is positive probability that all six coordinates are negative. Accordingly, the probability that $R$ has non-positive entries converges to some non-zero constant value as $s \rightarrow \infty$. 
\end{proof}

\begin{lemma}
\label{pair-lemma}

Define constant values
$$
\alpha_1 = 14 \cdot 2^{1/7} \cdot 3^{3/7} \approx 24.75, \qquad \qquad
\alpha_2 = 7 \cdot 2^{6/7} \approx 12.68
$$

There are $O( \frac{\alpha_1^s}{\sqrt{s}} )$ pairs $(x,y,z), (x', y', z') \in H$ with $[x = x'] + [y = y'] + [z = z'] \geq 1$.

 There are $O(  \frac{\alpha_2^s}{\sqrt{s}} )$ pairs $(x,y,z), (x', y', z') \in H$ with $[x = x'] + [y = y'] + [z = z'] \geq 2$.
\end{lemma}
\begin{proof}
For the first bound, we count the pairs $(x,y, z), (x', y', z') \in H$ with $x = x'$; the cases with the other two coordinates are completely symmetric.  Let $\ell = |x|$; since $(x,y,z) \in H$ we must have $\ell \leq 4s/7$. For each coordinate $i$ with $x_i = x'_i =  1$, the values $y_i,y'_i, z_i, z'_i$ can be arbitrary (giving $16^{\ell}$ choices); for each coordinate $i$ with $x_i = x'_i = 0$, there are three non-zero choices for $(y_i, z_i)$ and three non-zero choices for $(y'_i, z'_i)$. Overall, there are $16^{\ell} 9^{s - \ell}$ choices for the vectors $y, z, y', z'$. Summing over $\ell$,  the total contribution is
$$
\sum_{\ell = 0}^{\lfloor 4 s/7 \rfloor} \binom{s}{\ell} 16^{\ell} 9^{s -\ell}.
$$

Let $h(\ell) = \binom{s}{\ell} 16^{\ell} 9^{s - \ell}$ be the $\ell^{\text{th}}$ summand here; since $\ell \leq 4 s/7$, we can observe that $\frac{h(\ell)}{h(\ell-1)} = \frac{16 (s+1 - \ell)}{9} \geq 4/3$. Hence, the entire sum $\sum_{\ell=0}^{4s/7} h(\ell)$ is within a constant factor of $h(4s/7)$. In turn, by Stirling's formula, we can calculate $h(4s/7) \leq O( \alpha_1^s / \sqrt{s} )$. 

For the second bound, consider pairs $(x, y,z), (x',y',z') \in H$ with $x = x', y = y'$; the other coordinates are again completely symmetric. Let $\ell = |x \vee y|$; since $(x,y,z) \in H$ we must have $\ell \leq 6 s/7$. For each coordinate $i$ with $x_i \vee y_i = x_i' \vee y'_i = 1$, there are $4$ choices for $z, z'$ and $3$ choices for $x, y$; for each coordinate with $x_i \vee y_i = x'_i \vee y'_i = 0$, we have $z_i = z'_i = 1$. Summing over $\ell$ gives:
$$
\sum_{\ell = 0 }^{\lfloor 6s/7 \rfloor} \binom{s}{\ell} 12^{\ell}
$$

Again, consider the $\ell^{\text{th}}$ summand $h(\ell) = \binom{s}{\ell} 12^{\ell}$; since $\ell \leq 6 s/7$, we have $\frac{h(\ell)}{h(\ell-1)} = \frac{12 (s+1 - \ell)}{\ell} \geq 2$. Hence the sum is within a constant factor of $h(6s/7)$. In turn, by Stirling's formula, we get $h(6s/7) \leq O( \alpha_2^s/ \sqrt{s} )$.
\end{proof}

\begin{proposition}
\label{ggprop}
Suppose we define matrices $G^A, G^B, G^C$ using Iverson notation by:
\begin{align*}
G^A_{x,z} &= \Bigl[ |x| \leq 4 s/7  \wedge |z| \leq 4 s/7  \wedge |x \vee z| \leq 6 s/7 \Bigr] \\
G^B_{y,z} &= \Bigl[ |y| \leq 4s/7  \wedge |z| \leq 4s/7 \wedge |y \vee z| \leq 6s/7 \Bigr] \\
G^C_{x,y} &= \Bigl[ |x| \leq 4s/7 \wedge |y| \leq 4s/7 \wedge |x \vee y| \leq 6s/7 \Bigr]
\end{align*}

Then the real-valued algorithm satisfies $E[ \tilde C_{i,j}] = \Theta(  C_{i,j} 7^s / \psi_3 )$ and
$$
\mathbb V[\tilde C_{i,j}] = O \Bigl( \frac{ (\frac{\alpha_1^s}{\sqrt{s}} (n_1 + n_2) + n_1 n_2 \frac{\alpha_2^s}{\sqrt{s}}) C_{i,j}^2 + (\frac{\alpha_1^s}{\sqrt{s}} + (n_1 + n_2) \frac{ \alpha_2^s }{\sqrt{s} } + n_1 n_2 7^s) n_3 C_{i,j}^{(2)}}{n_1^2 n_2^2 n_3^2} \Bigr)
$$
\end{proposition}
\begin{proof}
Here $G_{xyz} = [ (x,y,z) \in H ]$. So $\sigma_0 = \sum_{(x,y,z) \in T} G_{x,y,z} = |H|$, which by Lemma~\ref{lem2} is $\Theta(7^s)$. Next, for $\sigma_1$, we compute
\begin{align*}
\sigma_1 &= n_1 \sum_{x } \Bigl( \sum_{y,z: (x,y,z) \in T} G_{xyz}  \Bigr)^2 =  \sum_x \Bigl( \sum_{y,z: (x,y,z) \in H} n_1 \Bigr)^2 = \sum_{\substack{ (x,y,z) \in H, (x',y',z') \in H \\ x = x'}} n_1
\end{align*}
By Lemma~\ref{pair-lemma}, this is $O( n_1 \alpha_1^s/\sqrt{s} )$. Similarly,  we have $\sigma_2 \leq O( n_2 \alpha_1^s / \sqrt{s}), \sigma_3 \leq O( n_3 \alpha_1^s/\sqrt{s})$,  and $\sigma_{12} \leq O( n_1 n_2 \alpha_2^s / \sqrt{s}), \sigma_{13} \leq O( n_1 n_3 \alpha_2^s / \sqrt{s}), \sigma_{23} \leq O( n_2 n_3\alpha_2^s/\sqrt{s} )$. Finally, we have $\sigma_{123} = n_1 n_2 n_3 \sum_{(x,y,z) \in T} G_{xyz}^2 = n_1 n_2 n_3 |H| \leq n_1 n_2 n_3  |T| = 7^s n_1 n_2 n_3$. Now apply \Cref{prop1}.
\end{proof}

\begin{theorem}
\label{cmatrixthm}
Let $A,B$ be non-negative real matrices. Let $\epsilon \in (0,1)$, and define parameters $\gamma = \psi_3 / \epsilon^2$ and $\beta_1 = \psi_3/\psi_1$ and $\beta_2 = \psi_3/\psi_2$.  With appropriate choice of parameters, we can obtain an estimated matrix $\tilde C$, such that, for any entry $i,j$, there holds
$$
\Pr(  \tilde C_{i,j}  \in [(1 - \epsilon) C_{i,j}, (1+\epsilon) C_{i,j}] ) \geq 3/4,
$$
with overall runtime
$$
O \Bigl(  \gamma^{\frac{\log 6}{\log 7}} + \gamma \bigl( (\beta_1 \sqrt{\log \beta_1})^{\frac{\log(6/7)}{\log(\alpha_1/7)}} + (\beta_2 \sqrt{\log \beta_2})^{\frac{\log(6/7)}{\log(\alpha_2/7)}} \bigr) \Bigr).
$$
\end{theorem}
\begin{proof}
If $\psi_3 \leq O( \psi_1)$ or $\psi_3 \leq O( \psi_2)$, then $\beta_1 \leq O(1)$ or $\beta_2 \leq O(1)$, and we can simply run the naive matrix multiplication algorithm with runtime $O( \psi_3 )$. Hence for the remainder of the proof, we assume that $\psi_3 \gg \psi_1$ and $\psi_3 \gg \psi_2$.

Our strategy will be to take multiple independent executions and average them so that the resulting sample means have relative variance $\epsilon^2/4$. By Chebyshev's inequality, the resulting sample means are then within $(1 \pm \epsilon) C_{i,j}$ with probability at least $3/4$. To that end, let us set
$$
s = \Big \lfloor \min\{ \log_7 \gamma, \log_{\alpha_1/7} (\beta_1 \sqrt{\log \beta_1}), \log_{\alpha_2/7} (\beta_2 \sqrt{\log \beta_2}) \} \Big \rfloor.
$$

We claim that $(\alpha_1/7)^s / \sqrt{s} \leq O( \psi_3 / \psi_1)$. For, since $s$ is chosen as the minimum, we have $(\alpha_1/7)^s/\sqrt{s} \leq (\alpha_1/7)^{s_1}/\sqrt{s_1 - 1}$ where $s_1 = \log_{\alpha_1/7}(\beta_1 \sqrt{\log \beta_1})$. Due to our assumption that $\psi_3 \gg \psi_1$, we may assume that $s_1 -1 \geq \log \beta_1$, and hence 
$$
(\alpha_1/7)^{s_1}/\sqrt{s_1 - 1} \leq (\beta_1 \sqrt{\log \beta_1}) / \sqrt{ \log \beta_1}  \leq O(\beta_1).
$$

By an analogous argument, we have also $(\alpha_2/7)^s / \sqrt{s} \leq O(\beta_2)$. By collecting terms in \Cref{ggprop}, and using the the fact that $C_{i,j}^{(2)} \leq C_{i,j}$, we get the estimate:
$$
\frac{ \mathbb V[\tilde C_{i,j}] }{\mathbb E[ \tilde C_{i,j}]^2} \leq O( \psi_3 / 7^s ).
$$

Thus, if we repeat the process for $\Omega(1 + \epsilon^{-2}/(7^s/\psi_3))$ iterations and take the mean of all observations, the overall relative variance of the resulting sample mean is reduced to $\epsilon^2/4$, and hence satisfies the required bounds. (We can scale by a known proportionality constant to get an unbiased estimator). The runtime is then $O( 6^s (1 + \epsilon^{-2}/(7^s/\psi_3)) )$ To bound this term, we can first estimate $6^s \leq 6^{\log_7 \gamma} = \gamma^{\frac{\log 6}{\log 7}}$,  and also we estimate
\begin{align*}
6^s \epsilon^{-2} / (7^s / \psi_3) &= (6/7)^s \psi_3 \epsilon^{-2} = \gamma \cdot (6/7)^{\min\{ \log_7 \gamma, \log_{\alpha_1/7}(\beta_1 \sqrt{\log \beta_1}), \log_{\alpha_2/7} (\beta_2 \sqrt{\log \beta_2}) \}} \\
&\leq \gamma \cdot \Bigl( (6/7)^{\log_7 \gamma} + (6/7)^{\log_{\alpha_1/7}(\beta_1 \sqrt{\log \beta_1})} + (6/7)^{\log_{\alpha_2/7}(\beta_2 \sqrt{\log \beta_2})} \Bigr),
\end{align*}
which after some simplifications, gives the claimed runtime bounds.
\end{proof}

\begin{corollary}
\label{cor11}
Suppose matrices $A,B$ are non-negative. With appropriate choice of parameters, the real-valued algorithm can compute $C$ with relative error $1 \pm \epsilon$ with probability $3/4$ in runtime
$$
O \Bigl( (\psi_3/\epsilon^2)^{\frac{\log 6}{\log 7}} + \frac{ \psi_1^{0.122} \psi_3^{0.878} + \psi_2^{0.259} \psi_3^{0.741}}{\epsilon^2} \Big).
$$

In particular, if the matrix is square of dimension $n$, then $\psi_1 = \Theta(n), \psi_2 = \Theta(n^2), \psi_3 = \Theta(n^3)$ and the runtime is $O( n^{2.77}/\epsilon^2 )$.
\end{corollary}

As usual, we can boost the success probability to any desired value $1 -\delta$ through median amplification, with a further $O( \log(1/\delta) )$ increase in runtime.  The analysis for the Boolean case is similar to the real-valued case. 

\begin{proposition}
\label{prop2}
If $C_{i,j} = 1$, then $\Pr( \tilde C_{i,j}  = 0)  \leq \exp \Bigl( -\Omega \bigl( \frac{7^s}{\psi_3} \bigr) +  O \bigl( \frac{ \psi_1 \alpha_1^s + \psi_2 \alpha_2^s }{\psi_3^2 \sqrt{s}} \bigr) \Bigr)$.
\end{proposition}
\begin{proof}
By \Cref{thm8c}, we have $\Pr( \tilde C_{i,j} = 0) = \Pr( \bigwedge_{t \in T} \overline{E_t} )$ for the events $E_t$. Note that we have $p_i = \Theta(1/n_i)$ for $i = 1,2,3$. We can get an upper bound $\Pr( \tilde C_{i,j} = 0) \leq \Pr( \bigwedge_{t \in H} \overline{ E_t } )$.  To calculate this probability via \Cref{janson-thm}, we use a standard estimate (see e.g. \cite{cite11}):
$$
\kappa = \sum_{t \in H} \frac{ \Pr(E_t)}{ \sum_{t' \in H, E_{t'} \sim E_t} \Pr( \overline{E_{t'}} \mid E_t) } \geq \sum_{t \in H} \Pr(E_t) - \frac{1}{2} \sum_{\substack{ t, t' \in H, t \neq t' \\ E_{t'} \sim E_t}} \Pr(E_{t'} \wedge E_t).
$$

Each event $E_t$ has probability $p_1 p_2 p_3 = \frac{1}{2 \psi_3}$. Since $|H| \geq \Omega(7^s)$, the first term is thus $\Omega(7^s/\psi_3)$.

For the second term, observe that for tuples $t = (x,y,z), t' = (x',y',z')$, we have $E_{t} \sim E_{t'}$ if and only if $x = x'$ or $y = y'$ or $z = z'$.  If $x = x', y \neq y', z \neq z'$, then $\Pr(E_t \wedge E_{t'}) = p_1 p_2^2 p_3^2$. By \Cref{pair-lemma}, the sum over all such pairs $t,t'$ is at most $\sum_{(x,y,z) \in H, (x,y',z') \in H} p_1 p_2^2 p_3^2 \leq O( \frac{\alpha_1^s}{n_1 n_2^2 n_3^2})$. Similarly, the pairs with $y = y', x \neq x', z \neq z'$ give a contribution of $O( \frac{\alpha_1^s}{n_1^2 n_2 n_3^2})$ and the pairs with $z = z', x \neq x', y \neq y'$ give a contribution of $O( \frac{\alpha_1^s}{n_1^2 n_2^2 n_3})$. These three cases add to $O( \frac{\alpha_1^s \psi_1}{\psi_3^2})$.  Similarly, the the terms where $t, t'$ overlap on two coodinates contribute at most $O( \frac{\alpha_2^s \psi_2}{\psi_3^2} )$.
\end{proof}

\begin{theorem}
Let $\delta \in (0,1)$, and define $\gamma = \psi_3 \log(1/\delta)$ and $\beta_1 = \psi_3/\psi_1$ and $\beta_2 = \psi_3/\psi_2$.  With appropriate choice of parameters, we can obtain an estimated matrix $\tilde C$, such that, for any entry $i,j$, there holds $\Pr(  \tilde C_{i,j}  = C_{i,j} ) \geq 1 - \delta$,
 with overall runtime
$$
O \Bigl(  \gamma^{\frac{\log 6}{\log 7}} + \gamma \bigl( (\beta_1 \sqrt{\log \beta_1})^{\frac{\log(6/7)}{\log(\alpha_1/7)}} + (\beta_2 \sqrt{\log \beta_2})^{\frac{\log(6/7)}{\log(\alpha_2/7)}} \bigr) \Bigr).
$$
\end{theorem}
\begin{proof}
Completely analogous to \Cref{cmatrixthm}, using the entrywise estimate from \Cref{prop2}.
\end{proof}

\begin{corollary}
With appropriate choice of parameters, we can compute the overall matrix $C$ correctly with probability at least $1 - 1/\poly(\psi_3)$ with runtime
$$
O \Bigl( (\psi_3 \log \psi_3)^{\frac{\log 6}{\log 7}} + \psi_1^{0.122} \psi_3^{0.878} + \psi_2^{0.259} \psi_3^{0.741} \Bigr).
$$
In particular, if the matrix is square of dimension $n$, then $\psi_1 = \Theta(n), \psi_2 = \Theta(n^2), \psi_3 = \Theta(n^3)$ and the runtime is $O( (n \log n)^{\frac{\log 6}{\log 7}})$.
\end{corollary}

\section{Concrete complexity estimates}
\label{sec:practical}
Since our goal is to get a new practical algorithm, the asymptotic estimates have limited value on their own. In this section, we calculate specific costing for problems which are at the upper limits of practical size. We want to compare our algorithm to Strassen's algorithm and the KK algorithm. 

Let us first consider how to measure the runtime. The KK algorithm will invoke $6^s$ recursive partitioning steps as given in \Cref{kk-thm1}. The cost of these step will have a quadratic component (adding and rearranging the relevant submatrices) as well as a near-cubic component (coming from the recursive subproblems). Assuming that the base case parameters are sufficiently large, the first of these should be negligible; overall, the work will be roughly $6^s W_d$ where $W_d$ is the base-case cost. By comparison, both Strassen's algorithm and the KK algorithm would run $s$ many recursive iterations, and then again pass to a base case, with costs of $7^s W_d$  and $6^s W_d$ respectively.

It is much trickier to estimate $W_d$ precisely. As a starting point, we might estimate $W_d \approx d_1 d_2 d_3$ in terms of arithmetic operations. But there are many additional factors to consider: a CPU with SIMD registers of width $w$ may be able to perform $w$ operations in parallel. Other optimizations may be available over $\text{GF}(2)$. Fortunately, in order to compare our algorithm with Strassen or KK, the precise value of $W_d$ does not matter; it is common to all the algorithms.

To simplify further, let us assume that $d_1 = d_2 = d_3$, and let $d$ denote this common value. For Strassen and KK, we then choose parameter $s = s_0 = \log_2(n/d)$. The runtimes of our algorithm, Strassen's algorithms, and the KK algorithm, are then respectively proportional to $6^s, 7^{s_0}, 6^{s_0}$.

At this point, let us bring up another important issue with costing the algorithms. As $s$ increases, the memory used by the algorithm also increases. Storing the original matrices requires roughly $n^2$ memory, and both Strassen and KK  algorithms use essentially this same amount of memory as well. Our algorithm, by contrast, may require $6^s$ memory to store the expanded $\bar A, \bar B$. Note that, for very large scale problems, the computational costs of the algorithm (i.e. total number of arithmetic operations) are dominated by communication and memory-bandwidth costs.  In order to keep this section relatively contained, we will not investigate further the memory costs.

\subsection{Case study: non-negative matrix multiplication}
Let us consider a scenario for real-valued matrix multiplication. For concreteness, let us suppose that $A,B$ are non-negative  and we want to estimate the value $C_{i,j}$ up to $(1 \pm \epsilon)$ for some constant value $\epsilon = 1/2$ and with success probability $3/4$. (This comparison is among the most favorable possible for our algorithm. Strassen's algorithm has error probability zero and has perfect accuracy.)

In our algorithm, by Chebyshev's inequality, we need to run for $\frac{4 \mathbb V[ \tilde C_{i,j}]}{\epsilon^2 \mathbb E[ \tilde C_{i,j}]^2 }$ trials, and the runtime will be  $6^s \cdot 16 \mathbb V[ \tilde C_{i,j}] / \mathbb E[ \tilde C_{i,j}]^2$. To compare with the other algorithms, we must choose parameter $s$ and we also need to determine the matrices $G^A, G^B, G^C$.

Although \cite{kk} did not discuss it, it is also possible to use KK for real-valued matrix multiplication (where we assume a random permutation of the rows and columns of $A,B$). We record the following resulting for their algorithm:
\begin{theorem}
Suppose that $C_{i,j}^{(2)} \leq C_{i,j}^2$. For the KK algorithm with $s$ iterations, any entry $i,j$ has
$$
\mathbb E[ \tilde C_{i,j} ] = (7/8)^s C_{i,j}, \qquad \mathbb V[ \tilde C_{i,j} ] \leq (7/8)^s C_{i,j}^2.
$$

In particular, it has relative variance $\mathbb V[ \tilde C_{i,j} ] / \mathbb E[ \tilde C_{i,j} ]^2 \leq (8/7)^s$.
\end{theorem}
\begin{proof}
Let $(x,y) \in V^2$ denote the randomized location of entry $i,j$. Then $|x \vee y| = s - L$, where $L$ is Binomial random variable with $s$ trials and probability $1/4$. Conditioned on $L$, the entry $\tilde C_{i,j}$ is a sum of $n_3 2^{-L}$ terms $A_{i,k} B_{k,j}$ where $k$ is chosen without replacement from $[n_3]$. So $\mathbb E[ \tilde C_{i,j} \mid L] = C_{i,j} 2^{-L}$ and we have
\begin{align*}
\mathbb E[ \tilde C_{i,j}^2 \mid L] &= \sum_{k = k'} A_{i,k} B_{k,j} A_{i,k'} B_{k',j} 2^{-L} + \sum_{k \neq k'} A_{i,k} B_{k,j} A_{i,k'} B_{k',j} 2^{-L} \cdot  \frac{n_3 2^{-L} - 1}{ n_3 - 1} \\
&\leq \sum_{k = k'} A_{i,k} B_{k,j} A_{i,k'} B_{k',j} (2^{-L} - 4^{-L}) + \sum_{k,  k'} A_{i,k} B_{k,j} A_{i,k'} B_{k',j} 4^{-L} \\
&= C_{i,j}^{(2)} (2^{-L}  - 4^{-L})+ C_{i,j}^2 4^{-L}.
\end{align*}

Integrating over $L$ gives $\mathbb E[ \tilde C_{i,j} ] = \sum_{\ell = 0}^s \binom{s}{\ell} (1/4)^{\ell} (3/4)^{s - \ell} \cdot C_{i,j} 2^{-\ell} = (7/8)^s C_{i,j}$
and
\begin{align*}
\mathbb E[ \tilde C_{i,j}^2 ] &= \sum_{\ell = 0}^s \binom{s}{\ell} (1/4)^{\ell} (3/4)^{s - \ell} \cdot (C_{i,j}^{(2)} (2^{-\ell} - 4^{-\ell}) + C_{i,j}^2 4^{-\ell}) \\
&= ( (7/8)^s - (13/16)^s )C_{i,j}^{(2)} + (13/16)^s C_{i,j}^2 \leq (7/8)^s C_{i,j}^2. \qedhere
\end{align*}
\end{proof}

Again by Chebyshev's inequality, the runtime for the KK algorithm is roughly $6^{s_0} \cdot 16 (8/7)^{s_0}$. Strassen's algorithm will take time $7^{s_0}$ (and have zero error).

Let us consider a square sample problem with $n = 2^{25}$, with square base case size $d = 2^7$. Note that storing the matrices already requires $2^{51}$ words of memory, which is at (or beyond) the upper practical limit. With these parameters, Strassen and KK will set $s_0 = 18$, and they have runtimes of respectively $2^{50.5}, 2^{54.0}$.

 We do not know how to select the $G$ matrices optimally, and the choice in \Cref{ggprop} would be suboptimal by large constant factors. We used a heuristic that the matrix $G^A_{x,z}$ should only depend on the Hamming weights of the vectors $x^{\base}, z^{\base}, x^{\base} \vee z^{\base}$; we also use a similar constraint for the other matrices $G^B, G^C$. This effectively reduces $G^A, G^B, G^C$ to $\Theta(s^3)$ degrees of freedom. Given this constraint, we use a local-search method to determine the matrices to minimize relative variance.

Figure~\ref{fig1} shows the relative variance and runtime for various choices of parameter $s$, and for our chosen accuracy parameter $\eps = 1/2$ and success probability $3/4$.

\begin{figure}[H]
\begin{center}
\begin{tabular}{|c||c||c|}
\hline
$s$ & $\mathbb V[\tilde C_{i,j}]/\mathbb E[ \tilde C_{i,j}]^2$ & \text{Runtime} \\
\hline
10 & $25.9$ & $55.8$ \\
11 & $23.1$ & $55.6$ \\
12 & $20.3$ & $55.4$ \\
13 & $17.6$ & $55.2$ \\
14 & $14.8$ & $55.0$ \\
15 & $12.1$ & $54.8$ \\
16 & $9.4$ & $54.8$ \\
17 & $6.9$ & $54.8$ \\
18 & $4.5$ & $55.0$ \\
19 & $2.4$ & $55.6$ \\
20 & $0.7$ & $56.4$ \\
21 & $-0.7$ & $57.6$ \\
\hline
\end{tabular}
\caption{Runtime of our algorithm for the sample problem. To handle the wide dynamic range, all figures are given in log base two; e.g., for $s = 10$, the relative variance is $2^{25.9}$ and the runtime is $2^{55.8}$}
\label{fig1} 
\end{center}
\end{figure}

The optimal value appears to be roughly $s \approx 15$, with runtime $2^{54.8}$. This is slightly worse than KK  algorithm, and much worse than Strassen's algorithm. We emphasize that this comparison is still unfair to Strassen's algorithm, which provides precise and deterministic calculations of $C_{i,j}$. Thus, despite the asymptotic advantages, it will probably never be profitable to our algorithm, or the KK algorithm, for real-valued matrix multiplication.

\subsection{Estimating performance for Boolean matrix multiplication}
Assuming we use the Boolean-to-$\text{GF}(2)$ reduction, all three algorithms (Strassen, KK, and ours) are randomized. Let us denote by $f$ the maximum failure probability for any entry $C_{i,j}$. In order to boost the algorithm to any desired small failure probability $\delta$, we need to repeat for $ \big \lceil \frac{\log(1/\delta)}{\log(1/f)} \big \rceil$  iterations. To put the algorithms on an even footing, we therefore define a cost parameter
$$
M = \frac{\text{Runtime}}{\log(1/f)}
$$

The extremal case for all the algorithms seems to be if there is exactly one witness $k$ with $A_{i,k}  B_{k,j} = 1$, in which case the KK algorithm has failure probability $f = 1 - (7/8)^s/2$ and Strassen's algorithm has failure probability $f = 1/2$. We thus have
$$
M_{\text{Strassen}} = 7^s/\log 2, \qquad \qquad M_{\text{KK}} = \frac{6^s}{-\log(1 - (7/8)^s/2)} \sim 2 \cdot (48/7)^s \approx 2 \cdot 6.857^s
$$

Now we need to estimate $f$, and hence $M$, for our algorithm.  Our starting point is to consider the random process of \Cref{thm8}, and to try to estimate the probability $f' = \Pr(U) \geq f$. There are two main approaches we can take here:  First, we can estimate $f'$ by Monte Carlo simulation. Second, we can use \Cref{janson-thm} while computing a (non-asymptotic) lower bound on parameter $\kappa$. 

Let us first discuss the Monte Carlo simulation. To implement it efficiently, we randomly draw the set of entries $x,y,z \in V$ with $I_{x,1} = 1$ or $I_{y,1} = 1$ or $I_{y,3} = 1$ respectively; then we loop over each such triples $(x,y,z)$ to check whether it is in $T$.   Since the expected number of non-zero values $I_{x,i}$ is $2^s p_i = \Theta(2^s d_i/n_i)$, the overall runtime for the simulation is roughly  $\frac{ 8^s d_1 d_2 d_3}{n_1 n_2 n_3 }$, much less than directly running the matrix multiplication algorithm. However, the Monte Carlo procedure only provides a reasonable estimate for medium-sized values of $f'$.

An alternative approach, which can work for any value of $f$, is to use a computable closed-form bound on $\kappa$. Note that this only provides an upper bound on $\Pr(U)$, not an unbiased estimate like the Monte Carlo simulation. We use the following formula:
\begin{proposition}
\label{comp-g-prop}
For a function $g: \{0, 1 \}^3 \rightarrow \mathbb N$, define the parameter $J_g$ by:
\begin{align*}
J_g &= 1 + p_1 (2^{s - g(100)} - 1) + p_2 (2^{s - g(010)}-1) + p_3 (2^{s-g(001)} - 1) \\
& \qquad + p_2 p_3 (4^s (3/4)^{g(001) + g(010) + g(011)} - 2^{s-g(010)} - 2^{s - g(001)} + 1) \\
& \qquad + p_1 p_3 (4^s (3/4)^{g(001) + g(100) + g(101)} - 2^{s-g(100)} - 2^{s-g(001)} + 1) \\
& \qquad + p_1 p_2 (4^s (3/4)^{g(100) + g(010) + g(110)} - 2^{s - g(100)} - 2^{s - g(010)} + 1).
\end{align*}

Then,
$$
\kappa \geq 7^s p_1 p_2 p_3 \sum_{g: \{0,1 \}^3 \rightarrow \mathbb N} \binom{s}{g(001), g(010), \dots, g(111)} / J_g.
$$
where the sum is over all possible functions $g$, and where we use the convention that the multinomial coefficient $\binom{s}{i_{001}, \dots, i_{111}}$ is zero unless $i_{001}, \dots, i_{111}$ are non-negative integers summing to $s$.
\end{proposition}
\begin{proof}
For any $t = (x,y,z) \in T$ we have $\Pr(E_t) = p_1 p_2 p_3$. Now consider some fixed $t$; for each value $w \in \{0,1 \}^3$,  define $g_t(w)$ to be the number of coordinates $i$ with $(x_i, y_i, z_i) = w$. So $g_u(000) = 0$ by definition of $T$ and $g_u(001) + \dots + g_u(111) = s$. 

We claim that $\sum_{t': E_{t'} \sim E_t} \Pr(E_{t'} \mid E_t) \leq J_{g_t}$. To illustrate, consider the case where $x' \neq x, y = y', z = z'$. Then $\Pr(E_{t'} \mid E_t) = p_1$. For each coordinate $i$ with $y_{i} = z_{i} = 0$, we must have $x_i = x'_i = 1$; otherwise, $x'_i = 1$ is unconstrained. Since there are precisely $g_t(100)$ such coordinates, we overall get $2^{s - g_t(100)}$ choices for $x'$. We need to subtract one to enforce $x \neq x'$, and so the total contribution to $\sum_{E_{t'} \sim E_t} \Pr(E_{t'} \mid E_t)$ from this case is $p_1 (2^{s - g_u(100) - 1})$. Similar arguments and calculations hold for other possibilites for how $t'$ intersects with $t$; we omit the calculations here.

At this point, we use the  following elementary observation for Janson's inequality (see \cite{cite11}):
$$
\kappa \geq \sum_{E} \frac{\Pr(E)}{\sum_{E' \sim E} \Pr(E' \mid E)}.
$$

Since $\Pr(E) = p_1 p_2 p_3$, this gives us the estimate $\kappa = \sum_{t \in T} p_1 p_2 p_3/J_{g_t}$. For any function $g = g_t$, there are precisely $\binom{s}{g(001) + \dots + g(111)}$ choices for the corresponding tuple $t \in T$.
\end{proof}

The sum in \Cref{comp-g-prop} has $\binom{s+6}{6}$ summands, and we can compute it in $\poly(s)$ time. This is much faster than actually running the matrix multiplication algorithm, or the Monte-Carlo simulation, either of which would have time exponential in $s$. 

We now consider a square sample problem with $n = 2^{25}$; following heuristics of \cite{cite10}, we take a square base case size $d = 2^{10}$. With these parameters, we can calculate
$$
M_{\text{Strassen}} = 2^{42.64}, M_{\text{KK}} = 2^{42.61}
$$
with both algorithms using $s_0 = 15$. To compare with our algorithm, we obtain estimates of failure probability by using empirical simulations with $10^6$ trials (denoted by $\hat \kappa = \log(1/\hat f')$) and well as using Proposition~\ref{comp-g-prop} (denoted by $\tilde \kappa$).  We denote the corresponding estimates of $M$ by $\tilde M$ and $\hat M$.

\begin{figure}[H]
\begin{center}
\begin{tabular}{|c||c|c|c||c|c|}
\hline
$s$ & $\mu$ & $\tilde \kappa$ & $\hat \kappa$ & $\tilde M$ & $\hat M$ \\
\hline
0 &  -46.0 & -46.0   & $$ & 46.0 & $$                     \\
1&     -43.2   & -43.2 & $$ & 45.8 & $$                \\
2& -40.4 & -40.4 & $$ & 45.6 & $$ \\
3&-37.6 & -37.6 & $$ & 45.3 & $$ \\
4& -34.8 & -34.8 & $$ & 45.1 & $$ \\
5 & -32.0 & -32.0 & $$ & 44.9 & $$ \\
6 & -29.2 & -29.2 & $$ & 44.7 & $$ \\
7 & -26.3 & -26.4 & $$ & 44.5 & $$ \\
8 & -23.5 & -23.6 & $$ & 44.2 & $$ \\
9 & -20.7 & -20.8 & -19.9 & 44.0 & 43.2 \\
10 & -17.9 & -18.0 & -18.9 & 43.8 & 44.8 \\
11 & -15.1 & -15.2 & -15.2 & 43.6 & 43.5 \\
12 & -12.3 & -12.5 & -12.5 & 43.5 & 43.5 \\
13 & -9.5 & -9.8 & -9.7 & 43.4 & 43.3 \\
14 & -6.7 & -7.3 & -7.1 & 43.5 & 43.2 \\
15 & -3.9 & -5.0 & -4.5 & 43.7 & 43.3 \\
16 & -1.1 & -2.9 & -2.1 & 44.3 & 43.5 \\
17 & 1.7 & -1.2 & -0.1 & 45.2 & 44.0 \\
18 & 4.5 & 0.2 & 1.6 & 46.4 & 44.9 \\
19 & 7.3 & 1.4 & 3.0 & 47.7 & 46.1 \\
20 & 10.1 & 2.5 & $$ & 49.2 & $$ \\
21 & 13.0 & 3.5 & $$ & 50.8 & $$ \\
\hline
\end{tabular}
\caption{Possible failure parameters for the algorithm. All figures are given in log base two. Some of the empirical estimates are left blank because no statistically significant estimates could be obtained. For sake of comparison with $\kappa$, we also write $\mu = 7^s p_1 p_2 p_3$.}
\label{fig2}
\end{center}
\end{figure}

Here, the resulting optimal value is roughly $s \approx 14$. The computed bound $\tilde \kappa$ is reasonably close to $\hat \kappa$ for medium values of $s$, including at the optimizing value $s = 14$. To select the parameter $s$, we would suggest minimizing the value of $\tilde M$ since it is very fast to compute. 

Overall, the runtime of this algorithm appears to be slightly worse than the KK algorithm. The key problem seems to be that there is still a significant  correlation among the appearance-patterns of triples in the matrix. Note that if the randomized algorithm could produce $7^s$ \emph{random} terms of the matrix tensor, the algorithm could significantly beat the KK algorithm; furthermore, \emph{asymptotically} the correlations are insignificant, in particular we have $\kappa \geq (1 - o(1)) \mu$. Unfortunately, this asymptotic behavior does not kick in for the problem sizes here. 

We also emphasize that there is no point to look at matrices which are larger than this case study; there, issues of memory size and bandwidth would become much more important than counting arithmetic operations. If our algorithm cannot win at for dimension $n \approx 2^{25}$, it will never win at \emph{any} problem scale in practice.

\section{Potential improvements}

In our algorithms, the hash functions $F_i$ are chosen uniformly at random. The number of preimages is thus a Binomial random variable with mean $\frac{2^s d_i}{n_i}$. There are likely better ways to choose the random hash functions. For example, we use a dependent rounding scheme to ensure that each preimage set has size exactly $\lfloor \frac{2^s d_i}{n_i} \rfloor$ or $\lceil \frac{2^s d_i}{n_i} \rceil$. Even go further, we could try to make the set of preimages be ``balanced'' in terms  of the Hamming weight of its elements.

Such modifications would likely reduce the variance of the statistical estimates. For practical parameter sizes, they could make a critical difference in performance. However, these modifications are also much harder to analyze precisely. We leave as an open question the benefit of such improvements and whether they can make this algorithm fully practical.

\section{Acknowledgments}
Thanks for Richard Stong for suggesting the proof of Lemma~\ref{lem2}. Thanks to conference and journal reviewers for helpful corrections and suggestions.

\end{document}